\documentclass{llncs}
\begin{document}
\title{Revenue Maximizing Markets for Zero-Day Exploits}
\author{Mingyu Guo\inst{1} \and Hideaki Hata\inst{2} \and Ali Babar\inst{1}}
\institute{
    School of Computer Science\\University of Adelaide, Australia\\
    \{mingyu.guo, ali.babar\}@adelaide.edu.au
    \and
    Graduate School of Information Science\\Nara Institute of Science
    Technology, Japan\\
    hata@is.naist.jp}
\maketitle

\begin{abstract} Markets for zero-day exploits (software vulnerabilities unknown to the vendor) have a long history and
    a growing popularity. We study these markets from a revenue-maximizing mechanism design perspective. We first
    propose a theoretical model for zero-day exploits markets. In our model, one exploit is being sold to multiple
    buyers. There are two kinds of buyers, which we call the defenders and the offenders.  The defenders are buyers who
    buy vulnerabilities in order to fix them ({\em e.g.}, software vendors).  The offenders, on the other hand, are
    buyers who intend to utilize the exploits ({\em e.g.}, national security agencies and police).  Our model is more
    than a single-item auction. First, an exploit is a piece of information, so one exploit can be sold to multiple
    buyers.  Second, buyers have externalities. If one defender wins, then the exploit becomes worthless to the
    offenders. Third, if we disclose the details of the exploit to the buyers before the auction, then they may leave
    with the information without paying. On the other hand, if we do not disclose the details, then it is difficult for
the buyers to come up with their private valuations. Considering the above, our proposed mechanism discloses the details
of the exploit to all offenders before the auction. The offenders then pay to delay the exploit being disclosed to the
defenders.\end{abstract}


\section{Introduction}

A zero-day exploit refers to a software vulnerability that has not been disclosed to the public, and is unknown to the
software vendor.  Information of new vulnerabilities gives cyber attackers free passes to attacking targets, while the
vulnerabilities remain undetected.  The trading of zero-day exploits has a long history, and selling them by
security researchers as ``legitimate source of income'' is a recent trend~\cite{Egelman:2013:MZE:2535813.2535818}.

Zero-day exploits markets are not necessarily black markets where the buyers are potential cyber criminals.  Software
vendors buy exploits via bug bounty reward programs.  National security agencies and police also buy exploits. 
% (Open,
% the valuable exploits target operating systems and browsers. 
It is widely reported that government agencies utilize
zero-day exploits to track criminals or for other national security reasons.  Financial industry companies buy exploits
to prevent attacks (once an exploit method is known, these companies can then carry out counter measures to
prevent attacks).  There are legitimate venture capital backed security companies whose business model is to sell
exploits for profit.  For example, ZeroDium is a zero-day acquisition firm, which buys high-risk vulnerabilities with
premium rewards, then resells them to mostly government clients~\cite{fisher2015threatpost}. Another similar company is
Vupen, who offers a subscription service for its clients, providing vulnerability data and exploits for zero days and
other bugs~\cite{fisher2015threatpost}.

Greenberg presented a price list of zero-day exploit sale, ranging from \$5,000-\$30,000 to
\$100,000-\$250,000~\cite{greenberg2012forbes}.
These prices are so high because it is generally difficult for the software vendor to independently discover these
vulnerabilities~\cite{Bilge:2012:BWK:2382196.2382284},
hence the exploits are expected to be alive for long periods of time.
%because, in general, data for analysis is not available until after an attack is discovered.
%Bilge and Dumitras reported that after vulnerabilities are disclosed publicly, the volume of attacks
%exploiting them increases by up to 5 orders of magnitude 

%Vulnerabilities are difficult to detect for software vendors,
%because, in general, data for analysis is not available until after an attack is discovered
%\cite{Bilge:2012:BWK:2382196.2382284}.

In this paper, we study markets for zero-day exploits from a revenue-maximizing mechanism design perspective. 
Our contributions are:

\begin{itemize}
    \item We present a theoretical mechanism design model for zero-day exploits markets.
        We identify the unique features of zero-day exploits markets.
        First, an exploit is a piece of information, so one exploit can be sold to multiple
buyers.  Second, buyers have externalities. We divide the buyers into two types: the offenders and the defenders. Once a
defender ``wins'' an exploit, the exploit becomes worthless for the offenders.
Third, if we disclose the details of the exploit to the buyers before the auction, then they may leave with the
information without paying. On the other hand, if we do not disclose the details, then it is difficult for the buyers to
come up with their private valuations. 

    \item We propose the {\em straight-forward (SF)} mechanism property, which requires that the mechanism discloses the
        full details of the exploit to the offenders before they submit their bids. In Proposition~\ref{prop:1}, we show
        that for the purpose of designing revenue-maximizing mechanisms, if SP, IR, and SF are required, then it is
        without loss of generality to focus on mechanisms that ``divide'' the time frame into two regions. Such a
        mechanism would disclose the full details of the exploit to all offenders, and then based on the reports from
        both the offenders and the defenders, pick an ending time. Before the ending time, the exploit is alive. Once it
        reaches the ending time, the exploit is revealed to the defenders, which renders it worthless. The offenders bid
        to keep the exploit alive, while the defenders bid to close the exploit earlier. Our model is similar to both
        the {\em cake-cutting} problem~\cite{Brams07:Better,Chen10:Truth} and the {\em single facility location}
        problem~\cite{Goemans04:Cooperative,Procaccia09:Approximate}.

    \item For a simplified single-parameter model, where every agent's type is characterized by a single parameter
        instead of a valuation function over time.  We modify and apply Myerson's classic technique for designing
        optimal single-item auction~\cite{Myerson81:Optimal} to our problem of dividing a continuous region.  We derive
        an optimal mechanism that maximizes the expected revenue for the single-parameter model.

    \item For the general model, we adopt the computationally feasible automated mechanism design
        approach~\cite{Guo10:Computationally}: instead of optimizing over all mechanisms, we focus on a family of
        parameterized mechanisms, and tune the parameters in order to obtain a good mechanism.  We focus on the AMA
        mechanisms used for designing revenue-maximizing combinatorial
        auctions~\cite{Likhodedov04:Boosting,Likhodedov05:Approximating}.  To
        identify a good AMA mechanism, we propose a technique that combines both optimization and heuristics.  We show
        via numerical experiments that our technique produces good revenue expectation: when applied to a
        single-parameter model (our technique does not require the single-parameter model), our technique achieves
        nearly $80\%$ of the optimal revenue (one reason we test our technique in a single-parameter model is that we are
able to calculate the optimal revenue for comparison).  \end{itemize}

\section{Model Description}

In this section, we present our mechanism design model for zero-day exploits markets. Our aim is to create a model with
minimal assumptions, and draw a parallel between our model and existing classic mechanism design models.

There is one exploit being sold to multiple game-theoretically strategic buyers. The seller is also the mechanism
designer, who wants to maximize her revenue ({\em e.g.}, the seller is a security company that sells exploits for
profit\footnote{Example such companies include ZeroDium and Vupen~\cite{fisher2015threatpost}.}).

\begin{assumption}
    The set of all buyers consists of two types of buyers: the defenders and the offenders.
    \begin{itemize}
        \item A defender is a buyer who buys exploits in order to fix them.
            Given a specific exploit, usually there is only one defender. For example, suppose the exploit attacks the Chrome browser, then Google
            is the defender, who would, for example, buy the exploit via its bug bounty reward program.
            Our model allows multiple defenders, but we assume that as soon as one defender gets hold of an exploit, the
            exploit gets immediately fixed, therefore rendering it worthless. That is, if one defender 
            receives information about an exploit, then all defenders benefit from it.
        \item An offender is a buyer who intends to utilize the exploit.
    \end{itemize}
\end{assumption}

Based on the above assumption, we cannot sell an exploit to an offensive buyer after we sell it to any defensive buyer.

\begin{assumption} One exploit is sold over a time frame $[0,1]$. $0$ represents the moment the exploit is ready for
    sale.  $1$ represents the exploit's end of life ({\em e.g.}, it could be the end of life of the affected software, or the
    release of a major service pack). 
    
    A mechanism outcome is represented by $(t_1,t_2,\ldots,t_n)$, where $t_i$ is the
    moment the exploit is disclosed to buyer $i$. 
    We assume that all defenders receive the information at the same time, denoted by $t_{end}$ (if one defender
    receives the information, then the exploit is fixed right away).
    It is also without loss of generality to assume that 
    if $i$ is an offensive buyer, then $t_i\in [0,t_{end}]$ (receiving the information after $t_{end}$ is equivalent to
    receiving it exactly at $t_{end}$ -- from this point on, the exploit is worthless).  

    Buyer $i$'s type is characterized by a nonnegative function $v_i(t)$. If $i$ is an offensive buyer who receives the exploit at 
    $t_i$ and the exploit gets fixed at $t_{end}$, then $i$'s valuation equals $\int_{t_i}^{t_{end}}v_i(t)dt$.
    Similarly, if $i$ is a defensive buyer, then her valuation equals $\int_{t_{end}}^{1}v_i(t)dt$.
    Basically, the offenders wish to keep the exploit alive for as long as possible, while the defenders wish to fix the
    exploit as early as possible.
\end{assumption}

A mechanism takes as input the valuation functions ($v_i(t)$ for $i=1,2,\ldots,n$), and produces an outcome $(t_1,t_2,\ldots,t_n)$ and a
payment vector $(p_1,p_2,\ldots,p_n)$, where $p_i$ is the amount $i$ pays under the mechanism. A buyer's utility equals
her valuation minus her payment.  We focus on mechanisms that are strategy-proof and individually rational.

\begin{definition} {\em Strategy-proof (SP)}: For every buyer $i$, by reporting $v_i(t)$ truthfully, her utility is maximized.\end{definition}

\begin{definition} {\em Individually rational (IR)}: For every buyer $i$, by reporting $v_i(t)$ truthfully, her utility is nonnegative. \end{definition}

Besides SP and IR, we introduce another mechanism property specifically for zero-day exploits markets.  

One thing we have been ignoring is that an exploit is a piece of information. As a result, if we disclose the exploit's
details to the buyers beforehand, then they may simply walk away with the information for free. If we do not describe
what we are selling, then it is difficult for the buyers to come up with their valuation functions.  

\begin{assumption}
    We assume that there are two ways for the seller to describe an exploit: either describe the full details, or
    describe what can be achieved with the exploit ({\em e.g.}, with this exploit, anyone can seize full control of a
    Windows 7 system remotely).
\begin{itemize}
    \item We assume that it is safe for the seller to disclose what can be achieved with the exploit.
        That is, the buyers will not be able to derive ``how it is done'' from ``what can be achieved''.
    \item If the seller only discloses what can be achieved, then it is difficult for an offensive
        buyer to determine whether the exploit is new, or something she already knows. It is therefore difficult for the
        offensive players to come up with their valuation functions in this kind of situation. They may 
        come up with expected valuation functions (by estimating how likely the exploit is new), but this may then lead
        to regret after the auction. 
    \item We assume that the defenders are able to come up with private valuation functions 
        just based on what can be achieved. This is because all zero-day exploits are by definition unknown to the
        defenders.
This assumption is consistent with practise. For bug bounty programs, vulnerabilities are generally classified into
different levels of severities, and vendors pay depending
on these classification~\cite{algarni2014software}. For example, Google Chrome provides guidelines to
classify vulnerabilities into critical, high, medium, and low severities, and pay accordingly~\cite{guideline2015chrome}.
    \end{itemize}
\end{assumption}

The above assumptions lead to the following mechanism property:

\begin{definition} {\em Straight-forward (SF)}: A mechanism is straight-forward if 
    the mechanism reveals the full details of the exploit to the offensive buyers, before asking for their valuation
    functions.
\end{definition}

It should be noted that SF does not require that the exploit details be revealed to the defenders before they bid.
If the seller does this, then the defenders would simply fix the exploit and bid $v_i(t)\equiv 0$. Due to IR, the
defenders can get away without paying.  

Offenders are revealed the details before they bid, but they cannot simply bid $v_i(t)\equiv 0$ to get away without 
paying. Our mechanisms' key idea is to use the defenders as ``threat''. That is, if the offenders bid 
too low, then we disclose the exploit to the defenders earlier, which renders the exploit worthless. Essentially,
the offenders need to pay to keep the exploit alive (the more they pay, the longer the exploit remains alive).

From now on, we focus on mechanisms that are SP, SF, and IR.
We present the following characterization result:

\begin{proposition}
    \label{prop:1}
Let $M$ be a mechanism that is strategy-proof, individually rational, and straight-forward.

We can easily construct $M'$ based on $M$, so that $M'$ is also strategy-proof, individually rational, and
straight-forward. $M'$ and $M$ have the same revenue for all type profiles. $M'$ takes the following form:

\begin{itemize}
\item At time $0$, the seller reveals the exploit in full details to all offenders, and reveals what can be achieved with the
exploit to all defenders.

\item Collect valuation functions from the buyers. 

\item Pick an outcome and a payment vector based on the reports. 
It should be noted that it is sufficient to represent the outcome using just $t_{end}$, which is when the exploit gets fixed.
\end{itemize}
\end{proposition}

The above proposition implies that for the purpose of design revenue-maximizing mechanisms, it is without loss of generality to focus on mechanisms with the above form.

\begin{proof}
    Given $M$, we modify it and construct $M'$ as follows: for all $i$ that is an offender, we move $t_i$ to $0$.
    For all $i$ that is a defender, we do not change $t_i$. For every type profile, we keep $M$'s payment vector. That
    is, $M'$ has the same revenue for every type profile. $M'$ is obviously SF.

    Now we show $M'$ is still SP and IR. It is easy to see that the defenders' valuations are not changed, so $M'$ is
    still SP and IR for the defenders.
    Offenders' valuations are changed. For offender $i$, originally under $M$, she receives the information at
    time $t_i$. Under $M'$, she receives the information at time $0$. It should be noted that because $M$ is SF, that
    means $t_i$ is not dependent on $i$'s own report. Therefore, the valuation increase for $i$, which equals
    $\int_0^{t_i}v_i(t)dt$, is independent of $i$'s own report. Hence, this increase of valuation does not change $i$'s
    strategy. $M'$ is still SP and IR for the offenders as well.
\end{proof}

\section{Comparing Against Classic Models}

To summarize our model, there are two types of agents (offenders and defenders).  Agent $i$'s type is characterized by
her valuation function $v_i(t)$.  The outcome $t_{end}(v_1,v_2,\ldots,v_n) \in [0,1]$ is chosen based on the type
profile.  The exploit is active between $[0,t_{end}]$, during which period all offenders can utilize the exploit.  The
exploit becomes worthless from $t_{end}$.  High bids (high valuation functions) from the offenders would push $t_{end}$
toward $1$, while high bids from the defenders would push $t_{end}$ toward $0$.  

Our model is very similar to both the {\em cake-cutting} problem and the {\em single facility location} problem.

{\em Cake-cutting:} The time frame $[0,1]$ can be viewed as the cake. $t_{end}$ cuts the cake into two halves. 
The agents' types are also characterized by valuation functions instead of single values. On the other hand, there are
also differences. For one thing, our model is more like {\em group} cake cutting, as both sides involve multiple agents.
Secondly, the offenders are bound to the left-hand side ($[0,t_{end}]$) while the defenders are bound to the
right-hand side ($[t_{end},1]$).  

{\em Single facility location:} $t_{end}$ can also be viewed as the position of the facility in a single
facility location problem.  The defenders are all positioned at $0$, so they prefer $t_{end}$ to be closer to $0$ (which
enlarges the interval $[t_{end},1]$). The offenders are all positioned at $1$, so they prefer $t_{end}$ to be closer to
$1$ (which enlarges the interval $[0,t_{end}]$).

Unfortunately, most cake-cutting and facility location literatures focus on money-free settings, so previous results
do not apply to our problem of revenue maximizing mechanism design.

\section{Optimal Single-Parameter Mechanism}

In this section, we study a simplified single-parameter model, and derive an optimal mechanism that maximizes the
expected revenue.  Results in this section are based on Myerson's technique on optimal single item auction, which
is modified to work for our problem.

\begin{assumption}
    Single-parameter model (we need this assumption only in this section): Agent $i$'s valuation function $v_i(t)$ is characterized by a single parameter $\theta_i \in
    [0,\infty)$:
        \[v_i(t) = \theta_i c_i(t)\]
        Here, $c_i(t)$ is a publicly known nonnegative function. That is, $i$'s type is characterized by a single parameter $\theta_i$.
\end{assumption}

For example, consider an offender $i$, if $c_i(t)$ represents the number of users $i$ may attack using the exploit at
time $t$, and $\theta_i$ is agent $i$'s valuation for attacking one user over one unit of time, then we
have $v_i(t) = \theta_i c_i(t)$. (For defenders, it'd be saving instead of attacking.)

For the single-parameter model, a mechanism is characterized by functions $t_{end}$ and $p$.
$t_{end}(\theta_1,\theta_2,\ldots,\theta_n)$ determines the outcome.  $p(\theta_1,\theta_2,\ldots,\theta_n)$ determines
the payment vector. Actually, for mechanisms that are SP and IR, $p$ is completely determined by the
allocation function $t_{end}$.

Fixing $\theta_{-i}$ and drop it from the notation, when agent $i$ reports $\theta_i$, we denote the outcome by
$t_{end}(\theta_i)$.

\begin{proposition}
    If $i$ is an offender, then we define \[x_i(\theta_i) = \int_{0}^{t_{end}(\theta_i)}c_i(t)dt\]

    If $i$ is a defender, then we define \[x_i(\theta_i) = \int_{t_{end}(\theta_i)}^{1}c_i(t)dt\]

    A mechanism is SP and IR if and only if for all $i$, $x_i(\theta_i)$ is nondecreasing in $\theta_i$, and
    agent $i$'s payment equals exactly

    \[\theta_ix_i(\theta_i)-\int_0^{\theta_i}x_i(z)dz\]
\end{proposition}

% The above proposition corresponds to Myerson's payment characterization result for single item auction. We have made the
% necessary adjustments so that it applies to our model. %For completeness, we still include the proof below.

\begin{proof}
    Suppose $x_i(\theta_i)$ is nondecreasing in $\theta_i$ and $i$ pays according to the above expression.
    By reporting $\theta_i$, $i$'s utility equals 
    \[\theta_ix_i(\theta_i)-\theta_ix_i(\theta_i)+\int_0^{\theta_i}x_i(z)dz = \int_0^{\theta_i}x_i(z)dz \ge 0.\]
    The above implies IR. We then show SP. By reporting $\theta_i'$, $i$'s utility equals
    \[\theta_ix_i(\theta_i')-\theta_i'x_i(\theta_i')+\int_0^{\theta_i'}x_i(z)dz.\]
    We subtract the above from $i$'s utility when reporting truthfully, the difference equals
    \[\int_0^{\theta_i}x_i(z)dz - \theta_ix_i(\theta_i')+\theta_i'x_i(\theta_i')-\int_0^{\theta_i'}x_i(z)dz.\]
    If $\theta_i>\theta_i'$, then the above equals
    \[\int_{\theta_i'}^{\theta_i}x_i(z)dz - (\theta_i-\theta_i')x_i(\theta_i') \ge
    \int_{\theta_i'}^{\theta_i}x_i(\theta_i')dz - (\theta_i-\theta_i')x_i(\theta_i')\]
    The right-hand side equals $0$. 
    Hence, under-reporting is never beneficial. Similarly, we can show over-reporting is
    never beneficial. 

    For the other direction, suppose the mechanism under discussion is SP and IR.  We use $p_i(\theta_i)$ to represent
    $i$'s payment. By SP, we have 
    \[\theta_ix_i(\theta_i)-p_i(\theta_i)\ge \theta_ix_i(\theta_i')-p_i(\theta_i')\]
    \[\theta_i'x_i(\theta_i')-p_i(\theta_i')\ge \theta_i'x_i(\theta_i)-p_i(\theta_i)\]
    Combining these two inequalities, we get 
    \[(\theta_i-\theta_i')x_i(\theta_i)\ge (\theta_i-\theta_i')x_i(\theta_i').\]
    Therefore, $x_i$ must be nondecreasing.

    By reporting $\theta_i'$, $i$'s utility equals $\theta_ix_i(\theta_i')-p_i(\theta_i')$.
    This is maximized when $\theta_i'=\theta_i$. Also, $\theta_i$ is arbitrary. We have $zx_i'(z)=p_i'(z)$.
    Integrating both sides from $0$ to $\theta_i$, we get that $i$'s payment must be as described in the
    proposition.
\end{proof}

For agent $i$, we assume $\theta_i$ is drawn independently from $0$ to an upper bound $H_i$, according to a probability density function
$f_i$ (and cumulative density function $F_i$). 
%Using Myerson's definition of virtual valuation, 
Agent $i$'s virtual
valuation $\phi_i(\theta_i)$ is defined as 
\[\phi_i(\theta_i) = \theta_i - \frac{1 - F_i(\theta_i)}{f_i(\theta_i)}\]
We need the {\em monotone hazard rate condition}: the virtual valuation
functions are nondecreasing (which is generally true for common distributions).

Given the payment characterization result, the expected payment from
agent $i$ equals $E_{\theta_i}(\phi_i(\theta_i)x_i(\theta_i))$. %(again, based on Myerson's technique). 
That is, 
given a type profile, to maximize revenue, we pick $t_{end}$ to maximize
$\sum_i(\phi_i(\theta_i)x_i(\theta_i))$. 
This decides how to pick the outcome. 

The last step is a new step on top of Myerson's technique, which is required
for our problem. For our model, $x_i$ is not necessarily bounded between
$0$ and $1$ (for single-item auction, the proportion won by an agent is between $0$ and $1$). Also, the sum of the $x_i$
is not necessarily bounded above by $1$ (for single-item auction, the total proportion allocated is at most $1$).
Without these bounds, picking the $x_i$ becomes more difficult. Fortunately, for our model, an outcome is characterized
by a single value, so we simply run a single dimensional optimization.  It should be noted that when an agent increases
her bid, her virtual valuation also increases according to the monotone hazard rate condition. This leads to higher
value for $x_i$ under our model. That is, the above rule for picking an outcome ensures that the $x_i$ are monotone.

The payments are then calculated according to the payment characterization result. The resulting mechanism maximizes the
expected revenue.

\section{General Model and Randomized Mechanisms}

In this section, we return to the original model where an agent's type is characterized by a valuation
function instead of a single parameter.

To design revenue-maximizing mechanisms for the general model, we adopt the computationally feasible automated mechanism
design approach~\cite{Guo10:Computationally}. That is, instead of optimizing over all mechanisms (which is too
difficult), 
we focus on a family of parameterized mechanisms, and tune the parameters in order to obtain a good
mechanism.  We focus on the AMA mechanisms used for designing revenue-maximizing combinatorial
auctions~\cite{Likhodedov04:Boosting,Likhodedov05:Approximating}. To identify an AMA mechanism with high revenue, we propose a technique that
combines both optimization and heuristic methods.

The family of AMA mechanisms includes the VCG mechanism as a special case. For our model, the VCG mechanism works as
follows:

\begin{itemize}
\item Pick an outcome $t_{end}^*$, which maximizes the agents' total valuation. We denote the set of offenders by $O$
    and the set of defenders by $D$. 
    For an offender $i\in O$, her valuation for outcome $t$
    equals $V_i(t)=\int_0^{t}v_i(z)dz$.
    For a defender $i\in D$, her valuation for outcome $t$
    equals $V_i(t)=\int_{t}^1v_i(z)dz$.

\[t_{end}^* = \arg\max_{t\in [0,1]}\{\sum_{i}V_i(t)\}\]
    
\item Then agent $i$ pays how much her presence hurts the other agents. That is, agent $i$ pays
\[\max_{t \in [0,1]}\{\sum_{j\neq i}V_j(t)\}-\sum_{j \neq i}V_j(t^*_{end})\]
\end{itemize}

The AMA mechanisms generalize the VCG mechanisms by assigning a positive coefficient $\mu_i$ to each agent.
The AMA mechanisms also assign an ``adjustment term'' $\lambda(o)$ for each outcome $o$, where $\lambda$ can be any
arbitrary function. For our model, the AMA mechanisms work as follows (different $\mu_i$ and $\lambda$ correspond to
different AMA mechanisms):

\begin{itemize}
\item Pick an outcome $t_{end}^*$, which maximizes the agents' total valuation, considering the $\mu_i$ and the function $\lambda$.
    \[t_{end}^* = \arg\max_{t\in [0,1]}\{\sum_{i}\mu_iV_i(t) + \lambda(t)\}\]
    
            \item Then agent $i$ pays how much her presence hurts the other agents, again, considering the $\mu_i$ and
                the function $\lambda$.
                Agent $i$ pays
                \[\frac{1}{\mu_i}\left(\max_{t \in [0,1]}\{\sum_{j\neq i}\mu_jV_j(t) + \lambda(t)\}-\sum_{j\neq i}\mu_jV_j(t^*_{end})-\lambda(t^*_{end})\right)\]
\end{itemize}

The idea behind the AMA mechanisms is that by assigning larger coefficients to the weaker agents (agents who most likely
lose according to the prior distribution), it increases competition, therefore increases revenue.
Also, if an outcome $o$ is frequently chosen and the agents have high surplus on this outcome, then by assigning a negative
$\lambda(o)$, the agents may be forced to pay more for this outcome.

All AMA mechanisms are SP and IR. Since we disclose the full details of the exploit to all offenders in the beginning,
SF is always guaranteed. 

So far, we have only considered deterministic mechanisms. $t^*_{end}$ refers to a particular moment. Between
$[0,t_{end}]$, the exploit is $100\%$ alive, while between $[t_{end},1]$, the exploit is $100\%$ dead (already fixed).
We could generalize the outcome space by allowing randomized mechanisms. A randomized mechanism's outcome is not just a
single value.
Instead, the outcome is characterized by a function $\alpha(t)$ over time.
For any moment $t$, $\alpha(t)$ represents the probability that the exploit is still alive at this moment.
$\alpha$'s values must be between $0$ and $1$, and it needs to nonincreasing. 
The new outcome space includes all deterministic outcomes. For example, the deterministic outcome $t^*_{end}$ is simply

\[\alpha(t)=\left\{
                       \begin{array}{ll}
                           1, & t\le t^*_{end} \\
                           0, & t > t^*_{end} 
                       \end{array} 
                   \right.
               \]

Allowing randomized mechanisms potentially increases the optimal expected revenue. For example, if one offender has extremely high
valuation with very low probability, then under a randomized mechanism, the mechanism could threat to
disclose the exploit with a low probability (say, $1\%$), unless the agent pays a buck load of money.
If the agent doesn't have high valuation, which is most of the time, then she wouldn't pay. 
Since the seller is only disclosing the exploit with $1\%$ probability, this does not change the expected revenue too much.
But if the agent does have high valuation, then the mechanism could earn way more from this agent.

A valid outcome function maps the time frame $[0,1]$ to values between $1$ and $0$, and are nonincreasing.  
Let $A$ be the outcome space.
It should be noted that $A$ does not have to contain all valid outcome functions.
Allowing randomization, the AMA mechanisms have the following form:

\begin{itemize}
    \item Pick an outcome function $\alpha\in A$, which maximizes the agents' total valuation, considering the $\mu_i$ and the function $\lambda$.

        For an offender $i\in O$, her valuation for outcome function $\alpha$
        equals $V_i(\alpha)=\int_0^{1}\alpha(z)v_i(z)dz$.

    For a defender $i\in D$, her valuation for outcome function $\alpha$
    equals $V_i(\alpha)=\int_{0}^1(1-\alpha(z))v_i(z)dz$.

    \[\alpha^* = \arg\max_{\alpha\in A}\{\sum_{i}\mu_iV_i(\alpha) + \lambda(\alpha)\}\]
    
            \item Then agent $i$ pays how much her presence hurts the other agents, again, considering the $\mu_i$ and
                the function $\lambda$.
                Agent $i$ pays
                \[\frac{1}{\mu_i}\left(\max_{\alpha \in A}\{\sum_{j\neq i}\mu_jV_j(\alpha) + \lambda(\alpha)\}-\sum_{j\neq
                i}\mu_jV_j(\alpha^*)-\lambda(\alpha^*)\right)\]
\end{itemize}

We need to pick the $\mu_i$ and $\lambda$ that correspond to high expected revenue. It is infeasible to numerically try
all $\mu_i$ values and all $\lambda$ functions.  As a result, we adopt a heuristic method for picking the $\mu_i$ and
$\lambda$.

First, we restrict the outcome space $A$ to functions of the following form:

\[\alpha(t)=\left\{
                       \begin{array}{ll}
                           \beta_1, & t\le \beta_2 \\
                           0, & t > \beta_2
                       \end{array} 
                   \right.
               \]
Here, both $\beta_1$ and $\beta_2$ are values between $0$ and $1$.
All functions in $A$ are characterized by these two parameters. We denote the outcome function characterized by $\beta_1$ and
$\beta_2$ by $\alpha_{\beta_1,\beta_2}$.
The idea is that instead of making the exploit $100\%$ alive from the beginning, we may simply kill the exploit right
from the beginning with probability $(1-\beta_1)$. 

We choose the following $\lambda$, where $\zeta$ is a parameter of the mechanism:
\[\lambda(\alpha_{\beta_1,\beta_2}) = \zeta(1-\beta_1)*\beta_2\]

What we are doing is that we reward outcomes that kill the exploit (with high probabilities) right from the beginning
(making these outcomes easier to get chosen under AMA). As a result,
if the offenders would like to keep the exploit alive with high probability from the beginning, they have to pay more. 
Previously, for deterministic mechanisms, the exploit is alive $100\%$ from the beginning. After the adjustments
here, the agents need to pay to achieve high probability from the beginning.

Once we focus our attention on $\lambda$ of the above form. An AMA mechanism is characterized by $n$ parameters: the
$\mu_i$ (except for $\mu_1$, since it is without loss of generality to set $\mu_1=1$) and $\zeta$.
For small number of agents, we are able to numerically optimize over these parameters and obtain an AMA mechanism with
good expected revenue.

\section{Example and Simulation}

In this section, we present an example mechanism design scenario, and simulate our proposed mechanisms' performances.

To make the examples more accessible, we consider a simple single-parameter setting involving just one offender (agent
$1$) and one defender (agent $2$).

For single-parameter settings, an agent's valuation function $v_i(t)$ equals $\theta_ic_i(t)$, where $c_i(t)$ describes
the pattern of this agent's valuation over time.
For the offender, we assume $c_1(t) = 1-t$. That is, the offender has higher valuation for the exploit earlier on, and
her valuation drops to $0$ at the end of the time frame. For the defender, we assume $c_2(t)=1$. That is, the defender's
valuation for the exploit does not change over time.

In order to make our example and simulation more realistic, we assume the exploit is a vulnerability of the Chrome
browser.  According to \cite{greenberg2012forbes}, an exploit that attacks the Chrome browser sells between $80k$ and
$200k$ for offensive clients (USD).  According to Google's official bug bounty reward program for the Chrome
browser~\cite{guideline2015chrome}, a serious exploit is priced between $0.5k$ and $15k$. That is, for a defender, we
expect the total valuation to be from this range.

The valuation of agent $1$ (the offender) for the exploit for the whole time frame equals $\theta_1\int_0^1(1-t)dt = \theta_1/2$.
So we assume $\theta_1$ is drawn from a uniform distribution $U(160,400)$.
The valuation of agent $2$ (the defender) for the exploit for the whole time frame equals $\theta_2\int_0^11dt = \theta_2$.
So we assume $\theta_2$ is drawn from a uniform distribution $U(0.5,15)$.

{\em Optimal single-parameter mechanism:} Agent $1$'s virtual valuation equals 
\[\phi_1(\theta_1) = \theta_1-\frac{1-\frac{\theta_1-160}{240}}{1/240} = 2\theta_1-400\]
Agent $2$'s virtual valuation equals 
%$\phi_1(\theta_1) = \theta_1-\frac{1-\frac{\theta_1-160}{240}}{1/240}$
%$\phi_1(\theta_1) = \theta_1-(400-\theta_1)$
Similarly, agent $2$'s virtual valuation equals $\phi_2(\theta_2) = 2\theta_2-15$.
Both are monotone as required.

Given a type profile, to maximize revenue, we pick $t_{end}$ to maximize
\[\sum_i(\phi_i(\theta_i)x_i(\theta_i)),\] 
where 
$x_1(\theta_1) = \int_{0}^{t_{end}}(1-t)dt = t_{end}-\frac{t_{end}^2}{2}$ and $x_2(\theta_2) = \int_{t_{end}}^{1}dt =
1-t_{end}$.
That is, we pick $t_{end}$ to maximize
\[(2\theta_1-400)(t_{end}-\frac{t_{end}^2}{2})+(2\theta_2-15)(1-t_{end})\]
For example, if $\theta_1=300$ and $\theta_2=10$, $t_{end}=0.975$. 

Based on the payment characterization result, agent $1$ pays $102.4$ and agent $2$ pays $0.2188$.
Considering all type profiles, the expected total revenue equals $79.20$.\\

{\em AMA mechanism:} As mentioned earlier, we focus on AMA mechanisms that are characterized by $2$ parameters: $\mu_2$
and $\zeta$. For each pair of parameters, we can simulate the expected revenue. After optimization, we choose $\mu_2=13$
and $\zeta=31$. For this pair, the expected revenue is $63.53$. This value is nearly $80\%$ of the optimal revenue
($79.20$). Also, we cannot achieve such good result without the heuristic term. If we set $\zeta=0$, then the obtained
revenue is $52.63$. 
We believe this example demonstrates the usefulness of our AMA and heuristic-based technique.
\\

{\em VCG mechanism:} The VCG mechanism is the AMA mechanism with $\mu2=1$ and $\zeta=0$. Under VCG, the
expected revenue is merely $7.667$.

\section{Conclusion}

In this paper, we study markets for zero-day exploits from a revenue-maximizing mechanism design perspective.  We
proposed a theoretical mechanism design model for zero-day exploits markets. By requiring a new mechanism property
called straight-forwardness, we also showed that for the purpose of designing revenue-maximizing mechanisms, it is
without loss of generality to focus on mechanisms that ``divide'' the time frame into two regions, which makes our model
similar to both the cake-cutting problem and the single facility location problem.

We first considered a simplified single-parameter model, where every agent's type is characterized by a single
parameter. With necessary modification and extension at the last step, we were able to apply Myerson's classic technique
for designing optimal single-item auction to our model and derived the optimal mechanism for single-parameter models.

For the general model, we adopted the computationally feasible automated mechanism design approach. We focused on the
AMA mechanisms.  To identify an AMA mechanism with high revenue, we proposed a technique that combines both
optimization and heuristics. Numerical experiments demonstrated that our AMA and heuristic-based technique performs
well.

\bibliographystyle{splncs03}
\bibliography{/home/mingyu/Dropbox/project_folder/mg}
\end{document}